\documentclass{article}
\usepackage{hchang}
\usepackage{hyperref}
\usepackage{comment}
\usepackage{tcolorbox}
\usepackage{color}
\usepackage{graphicx} % allow embedded images
  \setkeys{Gin}{width=\linewidth,totalheight=\textheight,keepaspectratio}
  \graphicspath{{graphics/}} % set of paths to search for images
\usepackage{booktabs} % book-quality tables
\usepackage{units}    % non-stacked fractions and better unit spacing
\usepackage{multicol} % multiple column layout facilities
\usepackage{lipsum}   % filler text
\usepackage{fancyvrb} % extended verbatim environments
  \fvset{fontsize=\normalsize}% default font size for fancy-verbatim environments
\usepackage{cleveref}
\usepackage{rotating}
\usepackage{transparent}
\usepackage{multicol}
\usepackage{tikz}
\usepackage{tcolorbox}
\usepackage{fullpage}
\usepackage{thmtools} 
\usepackage{thm-restate}
\usepackage{soul}
\usepackage{dashrule}

\usepackage{algorithmicx}
\usepackage{algpseudocode}
\usepackage{algorithm}
\usepackage{tikz}
\usetikzlibrary{decorations.pathreplacing,calc}
\algtext*{EndIf}
\algtext*{EndFor}
\algtext*{EndProcedure}

\usepackage{xcolor}
\newcommand{\lightgray}[1]{\textcolor{gray}{#1}}
\newcommand{\blue}[1]{{\color{blue}#1}}

\newtheorem{lemma}{Lemma}
\newtheorem{theorem}{Theorem}

\newtheorem{conjecture}{Conjecture}

\newcommand{\real}{\mathbb{R}}
\newcommand{\pr}{\mathrm{Pr}}

\newcommand{\cK}{\mathcal{K}}

\title{A Separator for Minor-Free Graphs Beyond the Flow Barrier}

\author{Hung Le\\ University of Massachusetts, Amherst}
\date{}

\begin{document}
\maketitle
\begin{abstract}

In 1990, Alon, Seymour, and Thomas~\cite{AST90} gave the first balanced separator of size $O(h^{3/2}\sqrt{n})$ for any $K_h$-minor-free graph, which has had numerous algorithmic applications. They conjectured that the size of the balanced separator can be reduced to $O(h\sqrt{n})$, which is asymptotically tight. Two decades later, Kawarabayashi and Reed~\cite{KR10} constructed a separator of size $O(h\sqrt{n} + f(h))$ based on the graph minor structure theorem, where $f(h)$ is an extremely fast-growing function typically seen in the structure theorem; their separator's size is only better than that of  Alon, Seymour, and Thomas for a very small value of $h$.  Recently, Spalding-Jamieson~\cite{Spalding25} constructed a separator of size $O(h\log h \log\log h \sqrt{n})$; the technique is rooted in concurrent flow-sparsest cut duality.   Spalding-Jamieson's separator comes very close to $O(h\log h \sqrt{n})$, which is the barrier for techniques based on the flow-cut duality.

In this work, we first observe that plugging in the recent padded decomposition by Filtser and Conroy~\cite{CF25} into the flow-based algorithm of Korhonen and Lokshtanov~\cite{KH24} yields a balanced separator of size $O(h\log h \sqrt{n})$, matching the flow barrier. This result motivates the question of whether the flow barrier can be broken, which would be a stepping stone toward resolving the conjecture of Alon, Seymour, and Thomas~\cite{AST90}. The main result of our work is a positive answer to this question: we construct a balanced separator of size $O(h \sqrt{\log h} \sqrt{n})$. Surprisingly, perhaps, our algorithm is still based on the iterative framework of Alon, Seymour, and Thomas~\cite{AST90}, although a key component of their algorithm within this framework, called the neighborhood bound, was shown to be tight~\cite{Kotlov00}. Our new idea is to incorporate a low-diameter decomposition into the framework, which allows us to reduce the neighborhood bound by a factor of $h$, at the cost of a factor $\log h$. As a result, we improve the $\sqrt{h}$ factor to $\sqrt{\log h}$ in the final separator's size. Erasing $\sqrt{\log h}$ entirely seems to require substantially new ideas.
    
\end{abstract}
\section{Introduction}

The two most powerful structural results in minor-free graphs that have numerous algorithmic applications are the graph minor structure theorem~\cite{RS03} and the separator theorem~\cite{AST90}. The structure theorem says roughly that one can decompose a $K_h$-minor-free graph into a tree structure of pieces which are simple, with the simplicity is parameterized by a function $f(h)$, an extremely fast-growing function~\cite{RS03}. A recent breakthrough~\cite{GSW25} reduces $f(h)$ to a polynomial function of $h$, though the proof remains very complicated, spanning hundreds of pages, and the degree of the polynomial remains impractically large. Specifically,  $f(h) = O(h^{2300})$.  Determining the exact polynomial dependency of $f(h)$ on $h$ (which could well be very practical) and/or giving a simpler construction remains a long-standing open problem~\cite{Lovsz2005}.  On the other hand, separator theorems often have very small dependence on $h$, and their proofs are often much simpler.

A \EMPH{balanced separator} of a graph $G = (V,E)$ is a subset $X\subseteq V$ such that the largest connected component of $G\setminus X$ has size at most $2n/3$. Alon, Seymour, and Thomas~\cite{AST90} were the first to show that any $K_h$-minor-free graph admits a balanced separator of size $O(h^{3/2}\sqrt{n})$, which can be found in polynomial time of $O(m\sqrt{nh})$.   Regarding the dependence on $h$ of the separator size, they remarked:
\begin{quote}
   ``We think that the expression $h^{3/2}n^{1/2}$ [...] is not the best possible, and that  $O(hn^{1/2})$ is the correct answer, but have not been able to decide this. [...] Every  3-regular expander with $n$ vertices is a graph with no $K_h$-minor for  $h = cn^{1/2}$, and with no separator of size $dn$, for appropriately  chosen positive constants $c$ and $d$; and hence the estimate $O(hn^{1/2})$  would be the best possible.''
\end{quote}

A few years later, Plotkin, Rao, and Smith~\cite{PRS94} constructed a balanced separator of size $O( h n\sqrt{\log n})$ which has an optimal dependence on $h$ but a larger dependence on $n$. This separator size is suboptimal in the regime $h = O(1)$, which includes well-studied classes, such as planar graphs and graphs with constant genus. Nevertheless, the bound by Plotkin, Rao, and Smith is strong evidence supporting the Alon-Seymour-Thomas conjecture.

\begin{conjecture}[Alon-Seymour-Thomas Conjecture~\cite{AST90}]\label{conj:separator} For any $h\geq 3$ and $n\geq 1$, any $K_h$-minor-free graph of $n$ vertices has a balanced separator of size at most $c \cdot h \sqrt{n}$ for some positive constant $c$.
\end{conjecture}

This conjecture saw very little progress for almost 20 years, until Kawarabayashi and Reed claimed a solution in a FOCS 2010 conference paper~\cite{KR10}.  However, as recently remarked by~\cite{BKLLM25}, one of the authors later clarified in a talk ~\cite{Kawara11} that the precise bound is $O(hn^{1/2} + f(h))$ where $f(h)$ is the function in the graph minor structure theorem. The full version of their paper remains unpublished.  And it is unclear if the recent breakthrough on the structure theorem~\cite{GSW25} can be applied to reduce $f(h)$ to a polynomial in $h$. Even if it does, then $f(h)$ remains a very large polynomial of $h$, which gives a vacuous bound when $h = n^{1/2300}$. Moreover, the construction would not be elementary, which has been regarded as the key advantage of separator theorems~\cite{AST90,PRS94,BKLLM25} over the structure theorem. In any case, \Cref{conj:separator} remains open in its full generality.

In a different direction, Biswal, Lee, and Rao~\cite{BLR10} studied spectral partitioning algorithms for $K_h$-minor-free graphs. Their technique gave a different algorithm for constructing a balanced separator of size $O(h^{3}\sqrt{n \log h})$ based on the duality between concurrent flow and (uniform) sparsest cut. While this separator bound is worse than that of  Alon, Seymour, and Thomas~\cite{AST90} by a factor of $h^{3/2}\sqrt{\log h}$, their framework made an important link between padded decomposition and balanced separators via concurrent flow. As of their writing, the best padding parameter was $\beta(h) = O(h^2)$. Recently, Spalding-Jamieson~\cite{Spalding25} studied reweighted spectral graph partitioning for special classes of graphs, and noticed that using the recent padded decomposition by Conroy and Filtser~\cite{CF25}, which has $\beta(h) = O(\log h)$, the separator bound can be improved to $O(h \log(h)\log\log h \sqrt{n})$ (Corollary 1.6. in~\cite{Spalding25}). We note that the proof by Spalding-Jamieson~\cite{Spalding25} via the spectral technique is somewhat complicated.  Here,  we give
a direct proof relating the padded decomposition parameter and the size of the balanced separator, and on the way, remove the $\log\log(h)$ factor from the separator bound by  Spalding-Jamieson~\cite{Spalding25}.

\begin{restatable}{theorem}{FlowSeparator}\label{thm:minor} For any $n$-vertex  $K_h$-minor-free graph, there exists a separator of size $O(h\log(h)\sqrt{n})$ that can be found in randomized polynomial time. 
\end{restatable}

Our proof of \Cref{thm:minor} is by adapting the proof of the separator theorem for $K_h$-induced minor-free graphs by Korhonen and  Lokshtanov~\cite{KH24}. In particular, we plug in the recent padded decomposition with a padding parameter $\beta(h) = O(\log h)$ by Conroy and Filtser~\cite{CF25} to replace $O(\log n)$ factor with $O(\log h)$ in the duality between concurrent flow and sparsest cut;  see \Cref{sec:flow} for details. 

The $\log(h)$ factor in \Cref{thm:minor} is precisely the padding parameter, which is also the gap between concurrent flow and (uniform) sparsest cut. Conceivably, if one can improve the gap, one gets a smaller separator. However, the $\log(h)$ gap is tight: for a degree-3 expander with at most $h^2$ vertices, which excludes $K_{O(h)}$ as a minor, has the flow-cut gap $\Omega(\log(h))$. Thus, the flow-based technique~\cite{BLR10,Spalding25,KH24} has a natural barrier of $\Omega(h \log h \sqrt{n})$, and \Cref{thm:minor} is optimal for this class of techniques.

In this paper, we set out to bypass the flow barrier. Our main result is a new combinatorial\footnote{Flow-based techniques use linear programming or other numerical methods, including SDP or spectral techniques, for recursively computing a sparse vertex cut, whose union is a balanced separator.} algorithm for constructing a balanced separator in the following theorem. 

\begin{table}[h]
\centering
\renewcommand{\arraystretch}{1.4} % Adjust this value for more/less vertical space
\begin{tabular}{r|r|l}
\hline
\textbf{Separator Size} & \textbf{References} & \textbf{Notes} \\ \hline
$O(h^{3/2} \sqrt{n})$ & \cite{AST90} & First balanced separator of size $O(\sqrt{n})$ for a fixed $h$ \\
$O(h \sqrt{n \log n})$ & \cite{PRS94} & Suboptimal by a factor of $\sqrt{\log n}$ for a fixed $h$ \\
$O(h \sqrt{n} + f(h))$ & \cite{KR10} & $f(h)$ is a fast growing function \\
$O(h^3 \sqrt{\log h} \sqrt{n})$ & \cite{BLR10} & Flow-cut duality \\
$O(h \log h \log \log h \sqrt{n})$ & \cite{Spalding25} & Flow-cut duality, using $O(\log h)$ padding parameter~\cite{CF25} \\
$O(h \log h \sqrt{n})$ & \Cref{thm:minor} &  Flow-cut duality, plugging~\cite{CF25} into~\cite{KH24}, matching flow barrier \\
$O(h \sqrt{\log h}\sqrt{n})$ &  \Cref{thm:Main} & Main result, bypassing the flow barrier \\
$O(h \sqrt{n})$ &  \Cref{conj:separator}~\cite{AST90} & Open \\\hdashline
$O(h^{13} \sqrt{n})$ & \cite{BKLLM25} & First linear-time algorithm for fixed $h$ \\
\end{tabular}
\caption{Results on balanced separators of $K_h$-minor-free graphs.}
\label{table:results}
\end{table}

\begin{theorem}\label{thm:Main}  For any $n$-vertex $K_h$-minor-free graph, there exists a separator of size $O(h \sqrt{\log h}\sqrt{n})$ that can be found in randomized polynomial time. 
\end{theorem}

 \Cref{thm:Main} cuts the gap between the flow-based techniques and \Cref{conj:separator} quadratically, from $\log(h)$ to $\sqrt{\log h}$. See \Cref{table:results} for an overview of existing bounds. Our proof is elementary, which we view as a stepping stone towards a full resolution of \Cref{conj:separator}.  Astute readers might realize that the dependence on $h$ in \Cref{thm:Main} is coincidentally the same as the dependence on $h$ in the maximum number of edges of a $K_h$-minor-free graph, which is $O(h\sqrt{\log h} n)$. The factor $h\sqrt{\log h}$ in the latter bound is tight~\cite{Kostochka82}, and realized by the Erdős–Rényi random graph $\mathcal{G}(n,p)$ with $p=1/2$~\cite{BCE80}. For $\mathcal{G}(n,p = 1/2)$, the clique minor size is $\Theta(n/\sqrt{\log n})$~\cite{BCE80}, so in this case, the separator bound $h \sqrt{n} = \Omega(n^{3/2}/\log(n))$ is vacuous. When $p = C/n$ for any constant $C > 1$, the size of the clique minor in $\mathcal{G}(n,p)$ is $\Theta(\sqrt{n})$ w.h.p~\cite{FKO08} and the separator size is $\Theta(n)$, which is also within the bound of \Cref{conj:separator}.  
 
\subsection{Other related work}

In this paper, we focus solely on the existence of a small balanced separator and make no attempt to minimize the running time. There is a long line of work on the computational aspect of balanced separators~\cite{AST90,PRS94,RW09,KR10,WulffNilsen11,wulffnilsen14,BKLLM25}. The goal is to design a faster algorithm, ideally in linear time, possibly at the cost of higher dependence on $h$ in the separator's size. The recent work by Bonnet, Korhonen, Le, Li, and Masařík~\cite{BKLLM25} gives the first linear-time algorithm when $h$ is a constant: the separator size is $O(h^{13}\sqrt{n})$ and the running time is $O(h^{13}n)$. We refer readers to their paper for an in-depth discussion of existing techniques for fast algorithms. A major open question is to obtain the best of both: a linear-time algorithm for computing a balanced separator of size $O(h\sqrt{n})$, which remains out of reach.

\subsection{Proof overview}\label{subsec:overview}

In this section, we give an overview of the proof of \Cref{thm:Main}. An important concept is a \EMPH{minor model}, denoted by $\mathcal{K}$, which is a collection of vertex-disjoint subgraphs, called \EMPH{branch sets}, such that every two branch sets $C, C' \in \mathcal{K}$ have an edge between them. The \EMPH{size} of the minor model is the number of branch sets it contains. A graph $G$ is $K_h$-minor-free if and only if it does not contain a minor model of size $h$.  We first describe an iterative algorithmic template for constructing balanced separators by Alon, Seymour, and Thomas~\cite{AST90} (though our presentation differs from the original by Alon, Seymour, and Thomas~\cite{AST90}). This template captures the algorithm by Alon, Seymour, and Thomas~\cite{AST90} and other variants~\cite{PRS94,WulffNilsen11}; see also~\cite{Le23}. Some components of the template will be redundant in one algorithm but are necessary in another. And all of them will be necessary in our final algorithm.

The algorithm iteratively maintains a minor model \EMPH{$\mathcal{K}$}, a set of vertices $\EMPH{$X$} \subseteq V$, and the largest connected component \EMPH{$H$} if one removes the ``current separator'' from the graph $G$. The current separator contains $X$ and some other vertices, denoted by \EMPH{$F(\mathcal{K}, H)$}; the precise definition of $F(\mathcal{K}, H)$ varies across algorithms. The key point is that it depends on the structure of $\mathcal{K}$ and $H$. Every branch set in $\cK$ has at least one neighbor in $H$. As long as $|V(H)| > 2n/3$ (which means the current separator is not balanced), the algorithm will shrink $H$ and add more vertices to either $X$ or $\mathcal{K}$, which effectively also updates the current separator. If at any point, the size of the minor model $\mathcal{K}$ reaches $h$, then it is a certificate that $G$ is not $K_h$-minor-free. 

Also, the template takes a parameter $\ell > 0$ as an input, and finds a separator whose size depends on $\ell$. In the two  algorithms most relevant to our work, namely, the Alon-Seymour-Thomas (AST) algorithm~\cite{AST90} and the Plotkin-Rao-Smith (PRS) algorithm~\cite{PRS94}, the separators' sizes are:
\begin{equation}\label{eq:AST-PRS}
    O\left(\frac{nh}{\ell} + h^2\ell\right) \quad \text{(AST algorithm \cite{AST90})} \qquad \&\qquad O\left(\frac{n}{\ell} + h^2\ell \log n\right) \quad \text{(PRS algorithm \cite{PRS94})} 
\end{equation}
In particular, choosing $\ell = \sqrt{n/h}$ in the former algorithm, and $\ell = \sqrt{n}/(h\sqrt{\log n})$ in the latter algorithm, we get separators of sizes $O(h^{3/2}\sqrt{n})$ and $O(h\sqrt{n\log n})$, respectively. 

In the AST algorithm~\cite{AST90}, in every step, either one can find (\EMPH{Case I}) a connected subgraph of $H$ with $O(h \ell)$ vertices that contains at least one neighbor of each branch set in $\cK$ or (\EMPH{Case II}) a subset $S$ of size $n/\ell$ such that at least one (small) connected component of $H\setminus S$ contains all the neighbors in $H$ of some branch set, say $C$, in $\mathcal{K}$. In the former case, it adds to $\mathcal{K}$ a new branch set of size $O(h \ell)$ and in the latter case, it extends an existing branch set $C\in \mathcal{K}$ such that in the next iteration $|N_{H}(C)| = O(n/\ell)$ where \EMPH{$N_{H}(C)$} is the set of neighbors of $C$ in $H$. The final separator, when the algorithm stops, contains for every branch set $C\in \cK$, either $V(C)$ or $N_H(C)$, whichever set has smaller size. (The AST algorithm does not maintain $X$.) 
Since $|\mathcal{K}| \leq h$, the former case contributes $O(h^2\ell)$, while the latter case contributes  $O(hn/\ell)$, to the separator's size in \Cref{eq:AST-PRS}. Interestingly, Kotlov~\cite{Kotlov00} shows that the trade-off $O(h \ell)$ in case(i) vs $O(n/\ell)$ in case (ii) of the AST algorithm is best possible. In other words, the \EMPH{neighborhood bound} $|N_{H}(C)| = O(n/\ell)$ is tight.

In the PRS algorithm~\cite{PRS94}, there are also two cases in every step. (Case i) if the radius of $H$ is at least $2\ell \ln(n)$, then one can add to $X$ a subset of vertices $S$ such that $|S| \leq \frac{|V(\Tilde{H})\setminus V(H)|}{\ell}$, where $\Tilde{H}$ is the largest component in the \underline{next iteration}.  (Case ii) Otherwise, the radius of $H$ is at most $2\ell \ln(n)$, then one can add to $\mathcal{K}$ a new branch set, which composes of at most $h$ shortest paths of length at most $2\ell \ln(n)$ each, where each path is connected to an existing branch set in $\mathcal{K}$. The final separator is $X\cup \left(\cup_{C\in \cK}V(C)\right)$.  To bound the size, in case i, the main observation is that $V(\Tilde{H})\setminus V(H)$ will not be involved in subsequent iterations, and therefore, by charging the size of $S$ to $V(\Tilde{H})\setminus V(H)$, every vertex is charged by $1/\ell$ amount, and it is charged at most once. This explains the first term $O(n/\ell)$ in \Cref{eq:AST-PRS}. In case ii, the branch set has size $O(h\ell \log n)$, and therefore, the total number of vertices in all (and at most $h$) branch sets is $O(h^2\ell \log n)$, which explains the second term in  \Cref{eq:AST-PRS}. Here we remark that the $\log(n)$ factor is inherent in the first step, since one has to bound the size of $S$ by vertices that are not in the largest component of the next step to avoid charging twice.

From the description above,  $F(\mathcal{K}, H) = \cup_{C\in \cK}V(C)$ in the PRS algorithm, while in AST algorithm, $F(\mathcal{K}, H) = \cup_{C\in \cK} \arg\min \{|V(C)|, |N_H(C)|\}$.

Our algorithm also fits the iterative algorithmic template described above and has four cases, corresponding to four steps. The first step is a new twist, which incorporates  the low-diameter decomposition (LDD) into the template. LDD says that one can remove a small number of vertices to get connected components of small weak diameter. Here the \EMPH{weak diameter} of a connected component $Y$ is $\max_{x\not= y \in Y} \delta_G(x,y)$, while the \EMPH{strong diameter} is $\max_{x\not= y \in Y} \delta_{G[Y]}(x,y)$.

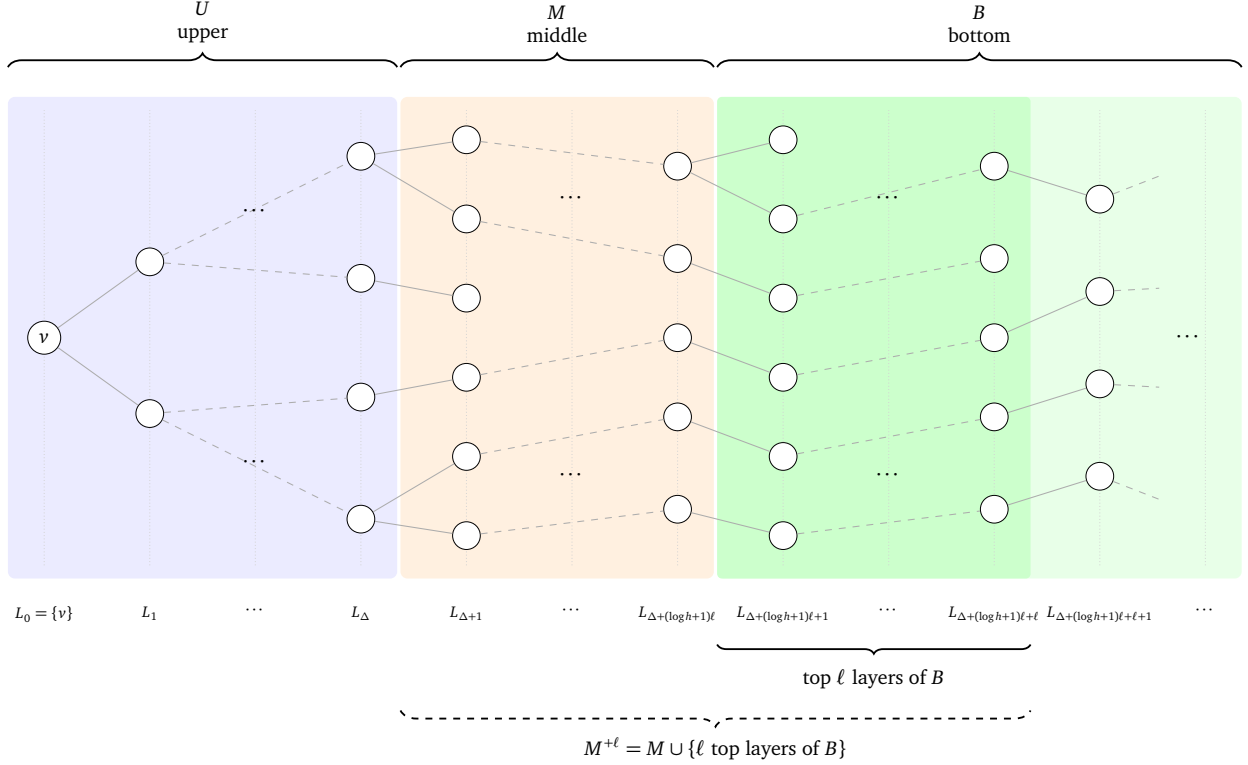
\begin{figure}[!htb]
    \centering
    \resizebox{1.0\textwidth}{!}{%  % Set to 80% of text width, '!' maintains aspect ratio
        \begin{tikzpicture}[
  every node/.style={font=\small},
  vertex/.style={circle,draw,fill=white,minimum size=4.2mm,inner sep=0pt},
  root/.style={circle,draw,fill=white,minimum size=5.0mm,inner sep=0pt,font=\small},
  treeedge/.style={draw=gray!65,line width=.45pt},
  omittededge/.style={treeedge,dashed},
  layerline/.style={draw=gray!35,densely dotted},
  setbrace/.style={decorate,decoration={brace,amplitude=6pt},thick},
  smallbrace/.style={decorate,decoration={brace,amplitude=4pt,mirror},thick},
  plusbrace/.style={decorate,decoration={brace,amplitude=7pt,mirror},thick,dashed},
  layerlabel/.style={font=\scriptsize,anchor=north,align=center}
]

% Shaded blocks for the three parts of the partition.
\fill[blue!8,rounded corners=3pt]    (-0.55,-3.65) rectangle (5.35,3.65);
\fill[orange!12,rounded corners=3pt] (5.40,-3.65) rectangle (10.15,3.65);
\fill[green!9,rounded corners=3pt]   (10.20,-3.65) rectangle (18.15,3.65);

% The top ell layers of B, included in M^{+\ell}.
\fill[green!20,rounded corners=2pt]  (10.20,-3.65) rectangle (14.95,3.65);

% Vertical layer guides.
\foreach \x in {0,1.6,3.2,4.8,6.4,8.0,9.6,11.2,12.8,14.4,16.0,17.6}{
  \draw[layerline] (\x,-3.45) -- (\x,3.45);
}

% Layer names.
\node[layerlabel] at (0,-3.95)    {$L_0=\{v\}$};
\node[layerlabel] at (1.6,-3.95)  {$L_1$};
\node[layerlabel] at (3.2,-3.95)  {$\cdots$};
\node[layerlabel] at (4.8,-3.95)  {$L_\Delta$};

\node[layerlabel] at (6.4,-3.95)  {$L_{\Delta+1}$};
\node[layerlabel] at (8.0,-3.95)  {$\cdots$};
\node[layerlabel] at (9.6,-3.95)  {$L_{\Delta+(\log h +1)\ell}$};

\node[layerlabel] at (11.2,-3.95) {$L_{\Delta+(\log h +1)\ell+1}$};
\node[layerlabel] at (12.8,-3.95) {$\cdots$};
\node[layerlabel] at (14.4,-3.95) {$L_{\Delta+(\log h +1)\ell+\ell}$};
\node[layerlabel] at (16.0,-3.95) {$L_{\Delta+(\log h +1)\ell+\ell+1}$};
\node[layerlabel] at (17.6,-3.95) {$\cdots$};

% Representative BFS tree vertices.
\node[root] (v) at (0,0) {$v$};

\node[vertex] (a1) at (1.6, 1.15) {};
\node[vertex] (a2) at (1.6,-1.15) {};

\node[vertex] (d1) at (4.8, 2.75) {};
\node[vertex] (d2) at (4.8, 0.90) {};
\node[vertex] (d3) at (4.8,-0.90) {};
\node[vertex] (d4) at (4.8,-2.75) {};

\node[vertex] (m1) at (6.4, 3.00) {};
\node[vertex] (m2) at (6.4, 1.80) {};
\node[vertex] (m3) at (6.4, 0.60) {};
\node[vertex] (m4) at (6.4,-0.60) {};
\node[vertex] (m5) at (6.4,-1.80) {};
\node[vertex] (m6) at (6.4,-3.00) {};

\node[vertex] (n1) at (9.6, 2.60) {};
\node[vertex] (n2) at (9.6, 1.20) {};
\node[vertex] (n3) at (9.6, 0.00) {};
\node[vertex] (n4) at (9.6,-1.20) {};
\node[vertex] (n5) at (9.6,-2.60) {};

\node[vertex] (p1) at (11.2, 3.00) {};
\node[vertex] (p2) at (11.2, 1.80) {};
\node[vertex] (p3) at (11.2, 0.60) {};
\node[vertex] (p4) at (11.2,-0.60) {};
\node[vertex] (p5) at (11.2,-1.80) {};
\node[vertex] (p6) at (11.2,-3.00) {};

\node[vertex] (q1) at (14.4, 2.60) {};
\node[vertex] (q2) at (14.4, 1.20) {};
\node[vertex] (q3) at (14.4, 0.00) {};
\node[vertex] (q4) at (14.4,-1.20) {};
\node[vertex] (q5) at (14.4,-2.60) {};

\node[vertex] (r1) at (16.0, 2.10) {};
\node[vertex] (r2) at (16.0, 0.70) {};
\node[vertex] (r3) at (16.0,-0.70) {};
\node[vertex] (r4) at (16.0,-2.10) {};

% Edges of the BFS tree T. Dashed edges cross omitted layers.
\draw[treeedge] (v) -- (a1) (v) -- (a2);

\draw[omittededge]
  (a1) -- (d1) (a1) -- (d2)
  (a2) -- (d3) (a2) -- (d4);

\draw[treeedge]
  (d1) -- (m1) (d1) -- (m2)
  (d2) -- (m3)
  (d3) -- (m4)
  (d4) -- (m5) (d4) -- (m6);

\draw[omittededge]
  (m1) -- (n1) (m2) -- (n2)
  (m4) -- (n3) (m5) -- (n4)
  (m6) -- (n5);

\draw[treeedge]
  (n1) -- (p1) (n1) -- (p2)
  (n2) -- (p3)
  (n3) -- (p4)
  (n4) -- (p5)
  (n5) -- (p6);

\draw[omittededge]
  (p2) -- (q1) (p3) -- (q2)
  (p4) -- (q3) (p5) -- (q4)
  (p6) -- (q5);

\draw[treeedge]
  (q1) -- (r1)
  (q3) -- (r2)
  (q4) -- (r3)
  (q5) -- (r4);

\draw[omittededge]
  (r1) -- ++(.90,.35)
  (r2) -- ++(.90,.05)
  (r3) -- ++(.90,-.05)
  (r4) -- ++(.90,-.35);

% Ellipses inside the tree where layers are suppressed.
\node at (3.2, 1.90) {$\cdots$};
\node at (3.2,-1.90) {$\cdots$};
\node at (8.0, 2.10) {$\cdots$};
\node at (8.0,-2.10) {$\cdots$};
\node at (12.8, 2.10) {$\cdots$};
\node at (12.8,-2.10) {$\cdots$};
\node at (17.35,0) {$\cdots$};

% Braces for U, M, and B.
\draw[setbrace] (-0.55,4.10) -- (5.35,4.10)
  node[midway,above=7pt,align=center]
  {$U$\\[-1pt] upper};

\draw[setbrace] (5.40,4.10) -- (10.15,4.10)
  node[midway,above=7pt,align=center]
  {$M$\\[-1pt]
   middle};

\draw[setbrace] (10.20,4.10) -- (18.15,4.10)
  node[midway,above=7pt,align=center]
  {$B$\\[-1pt] bottom};

% Brace for the top ell layers of B.
\draw[smallbrace] (10.20,-4.65) -- (14.95,-4.65)
  node[midway,below=7pt,align=center]
  {top $\ell$ layers of $B$};

% Brace for M^{+\ell}.
\draw[plusbrace] (5.40,-5.65) -- (14.95,-5.65)
  node[midway,below=8pt,align=center]
  {$M^{+\ell}
    =M\cup \{\text{$\ell$ top layers of $B$}\}$};

%\node[anchor=west,align=left] at (-0.55,5.55)
  %{BFS tree $T$ rooted at $v$; each vertical column is one BFS layer.};

\end{tikzpicture}
    }
    \caption{Each $L_i$ is the $i$-th layer of the BFS tree rooted at $v$, starting from $L_0 = \{v\}$.}
    \label{fig:UMB-vis}
\end{figure}

\begin{restatable}[Low Diameter Decomposition (LDD)]{lemma}{LDDLem}\label{cor:separation} Let $G = (V,E)$ be an undirected, unweighted $K_h$-minor-free graphs. Let  $\Delta > 0$ be a parameter. Then we can find in randomized polynomial time a subset $S\subseteq V$ of vertices of size $O(\log(h)\frac{n}{\Delta})$ such that w.h.p. every connected component $G\setminus S$ has a weak diameter at most $\Delta$. 
\end{restatable}

We note that the LDD in \Cref{cor:separation} is a simple corollary of padded decomposition~\cite{CF25}; see \Cref{sec:prelim}. Our algorithm has four steps:
\begin{itemize}
    \item (Step 1) we apply LDD to $H$ with $\Delta = \ell \log h$ to find a set $S$ of size $O(n/\ell)$. If every connected component of $H\setminus S$ has size less than $2n/3$ vertices, then we have found a balanced separator $X\cup S\cup F(\cK,H)$. Otherwise, we have found a vertex $v$ such that at least $2n/3$ other vertices are within distance $\Delta$ from $v$; note that we do not add any new vertex to the current separator $X$ in this case. We use $v$ to divide $V(H)$ into 3 parts, based on the distance from $v$: $U = \{u: \delta_H(u,v)\leq \Delta\}$ (upper), $M = \{u: \Delta < \delta_H(u,v)\leq \Delta + (\log h+1)\ell\}$ (middle), and $B = \{u: \delta_H(u,v)> \Delta + (\log h+1)\ell\}$ (bottom); see \Cref{fig:UMB-vis}. For a technical reason, we need to define an extension of $M$, denoted by $M^{+\ell}$, which contains $M$ and all vertices of $B$ within distance $\ell$ from $M$.  (These sets are continuous layers in the BFS tree rooted at $v$.)  The remaining three steps are based on the structure of the three parts.
    \item (Step 2) If every branch set $C\in \mathcal{K}$ has a neighbor in $U\cup M^{+\ell}$, then we can add a new branch set of size $O(h \ell \log(h))$ to $\cK$, consisting of at most $h$ shortest paths, each of length at most $\Delta  + (\log h+1)\ell + \ell = O(\ell \log h)$. 
    \item  (Step 3) If $|B| \leq n/h$, which means at least one layer among $\ell$ layers in  $M^{+\ell}\cap B$ has size at most $n/(\ell h)$.  Then we can extend an existing branch set $C\in \cK$ such that in the next iteration $|N_H(C)| \leq n/(h\ell)$. Compared to the neighborhood bound in Case II of the AST algorithm, we reduce the bound by a factor of $h$.  
    
    \item (Step 4) If $|B|>n/h$, then can find a set $S$ which is a layer $i^*$ in $M$ such that $|S| \leq \frac{|L_{\geq i^*+1}|}{\ell}$ where  $L_{\geq i^*+1}$ are all layers after $i^*$.  Then we add $S$ to $X$, and charge its size to the vertices in $L_{\geq i^*+1}$. By the LDD in Step 1, we know that at least $2n/3$ vertices are in $U$. Thus, in the next iteration, the largest connected component will not contain any vertex from  $L_{\geq i^*+1}$, and hence every vertex will be charged once during the entire algorithm. This means $|X| \leq n/\ell$ in the end.
\end{itemize}

If we never return the separator in Step 1, the final separator is $X\cup F(H,\mathcal{K})$; otherwise, it is $X\cup S\cup F(\cK,H)$ where $S$ is the vertex set from the LDD~\Cref{cor:separation}. Here $F(\mathcal{K}, H) = \cup_{C\in \cK} \arg\min \{|V(C)|, |N_H(C)|\}$, which is the same as in the AST algorithm. To bound the size, observe that in Step 1, $|S|  = O(n/\ell)$ by \Cref{cor:separation}. In Step 2, $|V(C)| \leq O(h\ell \log h)$ and in Step 3, $N_H(C) \leq n/(h\ell)$. As $\cK$ contains at most $h$ branch sets, the total size of $F(\cK,H)$ is $h\cdot O(h\ell \log h + n/(h\ell)) = O(n/\ell + h^2\ell  \log h)$. In Step 4, $|X| \leq n/\ell$, giving a separator of total size:
\begin{equation}\label{eq:our-bound}
    O(n/\ell + h^2 \ell \log h)
\end{equation}
 By choosing $\ell = \sqrt{n}/(h\sqrt{\log h})$, we get a separator of size $O(h\sqrt{\log h}\sqrt{n})$ as in \Cref{thm:Main}. As mentioned above, the root of our improvement is the smaller neighborhood bound resulting from the LLD step.

\section{Preliminaries}\label{sec:prelim}

Let $G = (V,E)$ be an undirected and unweighted graph. If the vertex set and edge set of $G$ are not explicitly written, then we denote the vertex set by $V(G)$ and the edge set by $E(G)$. If edges of $G$ are associated with a weight function $w$, then we write $G = (V,E,w)$. Graphs in this paper are generally unweighted, unless noted otherwise.  We denote by $\delta_G(u,v)$ the shortest distance in $G$ between any two vertices $u,v\in V$. Let $B_{G}(v,r) = \{u: \delta_G(v,u)\leq r\}$ be the ball of radius $r$ centered at  a vertex $v$.

Let $\Delta > 0$ be a parameter and an edge-weighted graph $G = (V,E,w)$. A \EMPH{$\Delta$-bounded} partition of $G$ is a partition $\mathcal{P}$ of $V$ such that the weak diameter of every set $P\in \mathcal{P}$ is at most $\Delta$. For a vertex $x$, we denote  by $\mathcal{P}(x)$ the set in $\mathcal{P}$ containing $x$.  Let $\beta \geq 1$ and $\delta \in (0,1)$ be two parameters. A \EMPH{$(\beta,\delta,\Delta)$-padded decomposition} of $G$  is a distribution of  $\Delta$-bounded partitions of $V$, denoted by $\mathcal{D}$, such that for any parameter $\gamma \in (0,\delta)$, we have:
\begin{equation}\label{eq:padd-prob}
    \pr_{\mathcal{P}\sim \mathcal{D}}[B_G(x,\gamma \Delta)\subseteq \mathcal{P}(x)]\geq  e^{-\beta \gamma}
\end{equation}

The parameter $\beta$ is called the \EMPH{padding parameter}. An ideal padded decomposition should have $\beta$ as small as possible, and $\delta$ close to $1$. We say that $G$ admits a \EMPH{$(\beta,\delta)$-padded decomposition scheme} if it admits a  $(\beta,\delta,\Delta)$-padded decomposition for any $\Delta > 0$. The scheme is \EMPH{samplable in polynomial time} if for every $\delta > 0$, one can sample a  $\Delta$-bounded partition from the $(\beta,\delta,\Delta)$-padded decomposition in polynomial time. Constructing a padded decomposition scheme for $K_h$-minor-free graphs with $\beta = O(\log h)$ and $\delta = \Omega(1)$ had been a long-standing open problem until the recent work by Conroy and Filtser~\cite{CF25} who constructed such a scheme.

\begin{theorem}[Padded Decomposition \cite{CF25}]\label{thm:padded-decomp} Every edge-weighted $K_h$-minor-free graph admits a $(\beta,\delta)$-padded decomposition scheme with $\beta = O(\log h)$ and $\delta = \Omega(1)$. Furthermore, the scheme is samplable in polynomial time. 
\end{theorem}

Now we show that the padded decomposition in \Cref{thm:padded-decomp} implies the LDD in \Cref{cor:separation}, restated below.

\LDDLem*
\begin{proof}
Let $\Delta$ be the parameter in the LDD. Let $\mathcal{P}$ be a $\Delta$-bounded partition sampled from a $(\beta = O(\log h), \gamma = 1/\Delta,\Delta)$-padded decomposition for $G$ (with unit edge weight). Let $S = \{s: B_G(s,1) \not\subseteq \mathcal{P}(s)\}$. Observe that:

\begin{equation}
    \begin{split}
       \mathbb{E}[|S|] &= \sum_{v\in V} \pr[B_G(v,1) \not\subseteq \mathcal{P}(v)] \\
       &\leq \sum_{v\in V}  1-e^{-\beta \gamma } \qquad \text{(by \Cref{eq:padd-prob})}\\
       &= \sum_{v\in V}  1-e^{-\beta/\Delta}  \qquad \text{(since $\gamma =1/\Delta
       $)}\\
       &\leq \sum_{v\in V} \beta/\Delta = O(n\log h/\Delta) \qquad \text{(since $\beta =O(\log h)
       $)}
    \end{split}
\end{equation}

To show that every connected component, say $C$, of $G\setminus S$ has weak diameter at most $\Delta$, we observe that $C \subseteq P$ for some set $P \in \mathcal{P}$, since otherwise, at least one vertex $v\in C$ has $B_G(v,1)\not\subseteq \mathcal{P}(v)$, and hence is in $S$, a contradiction.
\end{proof}

\section{Proof of \Cref{thm:Main}}

Our algorithm will take a parameter $\ell > 0 $ as an input and output a balanced separator of size depending on $\ell$, as stated in the following lemma.

\begin{restatable}{lemma}{mainLemma}\label{lm:Main} Let $G$ be an $n$-vertex graph $K_h$-minor-free graph, and $\ell > 0$ be any parameter. Then one can find a balanced separator of size  $O(\frac{n}{\ell} + \ell h^2\log h)$ in randomized polynomial time. 
\end{restatable}

Note that \Cref{thm:Main} follows by choosing  $\ell = \sqrt{n}/(h\sqrt{\log h})$. 

\subsection{The Algorithm}\label{subsec:the-alg}

Our algorithm follows the iterative framework described in \Cref{subsec:overview}; see \Cref{alg:separator} for the pseudocode. In each iteration, it maintains: (i) a clique minor model $\mathcal{K}$, (ii) a subset of vertices $X$, and (iii) a subgraph $H$ of $G$, which is the largest connected component of $G$ obtained by removing from $G$ the current balanced separator. The current balanced separator contains $X$ and a subset of vertices defined below, denoted by $F(\mathcal{K},H)$, which depends on  $\mathcal{K}$ and $H$. In every step, the algorithm will shrink $H$, update $\mathcal{K}$, and/or expand $X$. (The size of $\mathcal{K}$ could be increased or decreased in every step.)

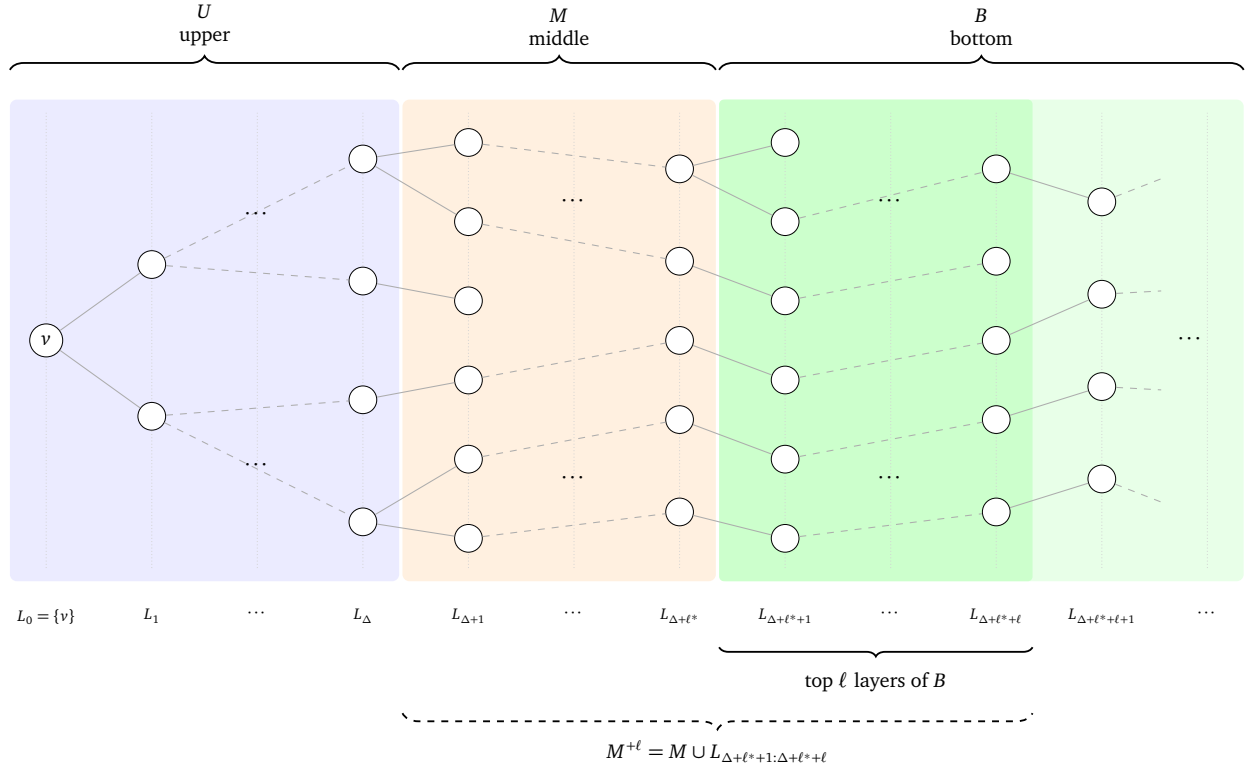
\begin{figure}[!htb]
    \centering
    \resizebox{1.0\textwidth}{!}{%  % Set to 80% of text width, '!' maintains aspect ratio
        \begin{tikzpicture}[
  every node/.style={font=\small},
  vertex/.style={circle,draw,fill=white,minimum size=4.2mm,inner sep=0pt},
  root/.style={circle,draw,fill=white,minimum size=5.0mm,inner sep=0pt,font=\small},
  treeedge/.style={draw=gray!65,line width=.45pt},
  omittededge/.style={treeedge,dashed},
  layerline/.style={draw=gray!35,densely dotted},
  setbrace/.style={decorate,decoration={brace,amplitude=6pt},thick},
  smallbrace/.style={decorate,decoration={brace,amplitude=4pt,mirror},thick},
  plusbrace/.style={decorate,decoration={brace,amplitude=7pt,mirror},thick,dashed},
  layerlabel/.style={font=\scriptsize,anchor=north,align=center}
]

% Shaded blocks for the three parts of the partition.
\fill[blue!8,rounded corners=3pt]    (-0.55,-3.65) rectangle (5.35,3.65);
\fill[orange!12,rounded corners=3pt] (5.40,-3.65) rectangle (10.15,3.65);
\fill[green!9,rounded corners=3pt]   (10.20,-3.65) rectangle (18.15,3.65);

% The top ell layers of B, included in M^{+\ell}.
\fill[green!20,rounded corners=2pt]  (10.20,-3.65) rectangle (14.95,3.65);

% Vertical layer guides.
\foreach \x in {0,1.6,3.2,4.8,6.4,8.0,9.6,11.2,12.8,14.4,16.0,17.6}{
  \draw[layerline] (\x,-3.45) -- (\x,3.45);
}

% Layer names.
\node[layerlabel] at (0,-3.95)    {$L_0=\{v\}$};
\node[layerlabel] at (1.6,-3.95)  {$L_1$};
\node[layerlabel] at (3.2,-3.95)  {$\cdots$};
\node[layerlabel] at (4.8,-3.95)  {$L_\Delta$};

\node[layerlabel] at (6.4,-3.95)  {$L_{\Delta+1}$};
\node[layerlabel] at (8.0,-3.95)  {$\cdots$};
\node[layerlabel] at (9.6,-3.95)  {$L_{\Delta+\ell^*}$};

\node[layerlabel] at (11.2,-3.95) {$L_{\Delta+\ell^*+1}$};
\node[layerlabel] at (12.8,-3.95) {$\cdots$};
\node[layerlabel] at (14.4,-3.95) {$L_{\Delta+\ell^*+\ell}$};
\node[layerlabel] at (16.0,-3.95) {$L_{\Delta+\ell^*+\ell+1}$};
\node[layerlabel] at (17.6,-3.95) {$\cdots$};

% Representative BFS tree vertices.
\node[root] (v) at (0,0) {$v$};

\node[vertex] (a1) at (1.6, 1.15) {};
\node[vertex] (a2) at (1.6,-1.15) {};

\node[vertex] (d1) at (4.8, 2.75) {};
\node[vertex] (d2) at (4.8, 0.90) {};
\node[vertex] (d3) at (4.8,-0.90) {};
\node[vertex] (d4) at (4.8,-2.75) {};

\node[vertex] (m1) at (6.4, 3.00) {};
\node[vertex] (m2) at (6.4, 1.80) {};
\node[vertex] (m3) at (6.4, 0.60) {};
\node[vertex] (m4) at (6.4,-0.60) {};
\node[vertex] (m5) at (6.4,-1.80) {};
\node[vertex] (m6) at (6.4,-3.00) {};

\node[vertex] (n1) at (9.6, 2.60) {};
\node[vertex] (n2) at (9.6, 1.20) {};
\node[vertex] (n3) at (9.6, 0.00) {};
\node[vertex] (n4) at (9.6,-1.20) {};
\node[vertex] (n5) at (9.6,-2.60) {};

\node[vertex] (p1) at (11.2, 3.00) {};
\node[vertex] (p2) at (11.2, 1.80) {};
\node[vertex] (p3) at (11.2, 0.60) {};
\node[vertex] (p4) at (11.2,-0.60) {};
\node[vertex] (p5) at (11.2,-1.80) {};
\node[vertex] (p6) at (11.2,-3.00) {};

\node[vertex] (q1) at (14.4, 2.60) {};
\node[vertex] (q2) at (14.4, 1.20) {};
\node[vertex] (q3) at (14.4, 0.00) {};
\node[vertex] (q4) at (14.4,-1.20) {};
\node[vertex] (q5) at (14.4,-2.60) {};

\node[vertex] (r1) at (16.0, 2.10) {};
\node[vertex] (r2) at (16.0, 0.70) {};
\node[vertex] (r3) at (16.0,-0.70) {};
\node[vertex] (r4) at (16.0,-2.10) {};

% Edges of the BFS tree T. Dashed edges cross omitted layers.
\draw[treeedge] (v) -- (a1) (v) -- (a2);

\draw[omittededge]
  (a1) -- (d1) (a1) -- (d2)
  (a2) -- (d3) (a2) -- (d4);

\draw[treeedge]
  (d1) -- (m1) (d1) -- (m2)
  (d2) -- (m3)
  (d3) -- (m4)
  (d4) -- (m5) (d4) -- (m6);

\draw[omittededge]
  (m1) -- (n1) (m2) -- (n2)
  (m4) -- (n3) (m5) -- (n4)
  (m6) -- (n5);

\draw[treeedge]
  (n1) -- (p1) (n1) -- (p2)
  (n2) -- (p3)
  (n3) -- (p4)
  (n4) -- (p5)
  (n5) -- (p6);

\draw[omittededge]
  (p2) -- (q1) (p3) -- (q2)
  (p4) -- (q3) (p5) -- (q4)
  (p6) -- (q5);

\draw[treeedge]
  (q1) -- (r1)
  (q3) -- (r2)
  (q4) -- (r3)
  (q5) -- (r4);

\draw[omittededge]
  (r1) -- ++(.90,.35)
  (r2) -- ++(.90,.05)
  (r3) -- ++(.90,-.05)
  (r4) -- ++(.90,-.35);

% Ellipses inside the tree where layers are suppressed.
\node at (3.2, 1.90) {$\cdots$};
\node at (3.2,-1.90) {$\cdots$};
\node at (8.0, 2.10) {$\cdots$};
\node at (8.0,-2.10) {$\cdots$};
\node at (12.8, 2.10) {$\cdots$};
\node at (12.8,-2.10) {$\cdots$};
\node at (17.35,0) {$\cdots$};

% Braces for U, M, and B.
\draw[setbrace] (-0.55,4.10) -- (5.35,4.10)
  node[midway,above=7pt,align=center]
  {$U$\\[-1pt] upper};

\draw[setbrace] (5.40,4.10) -- (10.15,4.10)
  node[midway,above=7pt,align=center]
  {$M$\\[-1pt]
   middle};

\draw[setbrace] (10.20,4.10) -- (18.15,4.10)
  node[midway,above=7pt,align=center]
  {$B$\\[-1pt] bottom};

% Brace for the top ell layers of B.
\draw[smallbrace] (10.20,-4.65) -- (14.95,-4.65)
  node[midway,below=7pt,align=center]
  {top $\ell$ layers of $B$};

% Brace for M^{+\ell}.
\draw[plusbrace] (5.40,-5.65) -- (14.95,-5.65)
  node[midway,below=8pt,align=center]
  {$M^{+\ell}
    =M\cup L_{\Delta+\ell^*+1:\Delta+\ell^*+\ell}$};

%\node[anchor=west,align=left] at (-0.55,5.55)
% {BFS tree $T$ rooted at $v$; each vertical column is one BFS layer.};

\end{tikzpicture}
    }
    \caption{Each $L_i$ is the $i$-th layer of the BFS tree rooted at $v$, starting from $L_0 = \{v\}$. Here $\Delta = \ell \log$ and $\ell^* = (\log h + 1)\ell$.}
    \label{fig:UMB-vis2}
\end{figure}

\hdashrule{0.95\linewidth}{1pt}{2pt}

\paragraph{Step 1: Low-diameter Decomposition.} Let $\Delta = \ell \log(h)$. We apply the LDD in \Cref{cor:separation} to $H$ with parameter $\Delta$ to find a subset $S\subseteq V(H)$ of size $|S| = O(\log(h) n/\Delta) = O(n/\ell)$ such that every connected component of $H\setminus S$  has a weak diameter at most $\Delta$. If every connected component has size at most $2n/3$, then we output $X\cup S\cup F(\mathcal{K})$ as a balanced separator. (Note here that $S$ has size $O(n/\ell)$, which is well within the bound in \Cref{lm:Main}; the size of $X$ and $F(\mathcal{K},H)$ will be inductively guaranteed to satisfy \Cref{lm:Main} as well.) 

Otherwise, pick a vertex $v$ such that $|B_H(v,\Delta)|\geq 2n/3$; such a vertex exists since at least one connected component of $H\setminus S$ has size at least $2n/3$. Let $T$ be the BFS rooted at $v$ of $H$. Let $L_i$ be the $i$-th layer of the BFS tree, starting from the $0$-th layer, which contains $v$ only. For any integers $0 \leq a\leq b$, let $\text{\EMPH{$L_{a:b}$}} = \{L_i: a\leq i \leq b\}$, $\text{\EMPH{$L_{\geq b}$}} = \cup_{i\geq b} L_i$ and $\text{\EMPH{$L_{\leq a}$}} = L_{0:a}$.  We partition the vertex set of $H$ into three sets:  $U =L_{\leq \Delta}$ (upper), $M = L_{\Delta + 1: \Delta + \ell^*}$ (middle), and $B = L_{\geq \Delta + \ell^* + 1}$ (bottom), where $\text{\EMPH{$\ell^*$}} =  (\log h+1)\ell$. Let $M^{+\ell} = M\cup L_{\Delta + \ell^* + 1:  \Delta + \ell^* + \ell}$. Note that $M^{+\ell}$ contains exactly $\ell$  top layers of $B$; see \Cref{fig:UMB-vis2}. We proceed to \emph{one of the following steps}. In every step $H$ will shrink.

\paragraph{Step 2: Growing the minor model.} This step applies when  $N_H(C)\cap (U\cup M^{+\ell})\not= \varnothing$ for every branch set $C\in \mathcal{K}$, where  $N_H(C)$ is the set of neighbors of $C$ in $H$. That is, every branch set $C\in \mathcal{K}$ has a neighbor in the first $O(\ell \log h)$ layers of $T$. In this case, we can add a new branch set $\Tilde{C}$ to $\mathcal{K}$, which consists of at most  $|\mathcal{K}|$  shortest paths, each from $v$ to a neighbor of a branch set in $\mathcal{K}$. We then remove $\Tilde{C}$ from $H$, and start a new iteration (without running the following subsequent steps).  Observe $\Tilde{C}$ has size $O(h\ell \log h)$, which is within the bound of \Cref{lm:Main}.

\paragraph{Step 3: Growing a branch set.~} If Step 2 does not apply, then there must exist a branch set $C$ such that $N_H(C)\subseteq B\setminus M^{+\ell}$. If $|B|\leq n/h$, then we grow $C$ into $H$ such that in the next iteration, $N_H(C) \leq \frac{n}{h\ell}$. 

More precisely, let $Y$ be the layer of minimum size among all the layers in $M^{+\ell}\setminus M$, and $y$ be its layer index, i.e., $Y = L_{y}$. Since $M^{+\ell}\setminus M$ is a subset of $B$ and has exactly $\ell$ layers, $|Y| \leq |B|/\ell = n/(h\ell)$.  Let $W = H[L_{\geq (y+1)}]$. Note that $W$ could be disconnected.  Let $Z$ be all the vertices reachable from $N_H(C)$ in the graph $W$. Then we grow $C$ by setting $C\leftarrow C\cup Z$, and shrink $H$ by removing all vertices of $Z$ from $H$. After this step, $N_H(C)\subseteq L_{y}$ and hence  $N_H(C) \leq |Y| \leq \frac{n}{h\ell}$. After removing $Z$ from $H$, a branch set in $\mathcal{K}$ may no longer have a neighbor in $H$; in this case, we \EMPH{trim $\mathcal{K}$} by removing every such branch set from $\mathcal{K}$. Then we start a new iteration.

Now we define:
\begin{equation}\label{eq:FKH-de}
    F(\mathcal{K}, H) \coloneqq \cup_{C\in \mathcal{K}}(\arg \min\{|V(C)|, |N_H(C)|\})
\end{equation}
That is, for every branch set $C$ in $\mathcal{K}$, $ F(\mathcal{K}, H)$ either contains $V(C)$ or its set of neighbors in $H$, whichever set has smaller size. The idea here is that one can disconnect $C$ from $H$ by deleting either $C$ itself or its set of neighbors in $H$; we prefer deleting the smaller set.  As mentioned in Step 1, we would like to have $|F(\mathcal{K}, H)| = O(n/\ell + \ell h^2\log h)$ (the same bound in \Cref{lm:Main}).  When a branch set $C$ is added to $\mathcal{K}$ (in Step 2), it has size $O(h\ell \log h)$. However, $C$ might be grown further in Step 3. If it never grows, then all such branch sets will contribute at most $|\mathcal{K}|\cdot O(h\ell \log h) = O(\ell h^2\log h)$ to $F(\mathcal{K}, H)$. If it does grow in Step 3 (even many times), then $|N_H(C)| \leq n/(h\ell)$ and therefore, the neighbor sets of such branch sets contribute $|\mathcal{K}| n/(h\ell) = O(n/\ell)$ to  $|F(\mathcal{K}, H)|$. Thus, in any case, $|F(\mathcal{K}, H)| = O(n/\ell + \ell h^2\log h)$. 

\paragraph{Step 4: Growing $X$.~} If Step 3 does not apply, then at this step, there must exist a branch set $C$ such that $N_H(C)\subseteq B\setminus M^{+\ell}$ and $|B|> n/h$. By a simple counting argument, we can show that there exists a layer $i^*\in [\Delta+1, \Delta + \ell^{*}]$ such that $L_{i^*} \leq L_{\geq i^*+1}/\ell$ (see \Cref{lm:importantlayer}). We then add $L_{i^*}$ to $X$, remove $L_{\geq i^*}$ from $H$, trim $\mathcal{K}$ (by removing branch sets that have no neighbor in $H$), and start a new iteration. 

Note that $X$ only grows in Step 4, and its size can be bounded by a charging argument. Since $L_{i^*} \leq L_{\geq i^*+1}/\ell$, we can charge the size of $L_{i^*}$ to vertices in $L_{\geq i^*+1}$; each vertex gets $1/\ell$ charge. Since vertices in $L_{\geq i^*}$ are removed from $H$, no vertex will be charged twice during the entire course of the algorithm, giving  $|X|  \leq n/\ell$.

\hdashrule{0.95\linewidth}{1pt}{2pt}
\\

Now we sketch the argument for bounding the size of the separator; the full analysis of the algorithm is given in the next section. As quickly argued in Step 3 and Step 4, it holds that $|F(\mathcal{K}, H)| = O(n/\ell + \ell  h^2\log h)$ and $|X| \leq n/\ell$. If the algorithm returns a separator in Step 1, then the total size is $|S| + |X| + |F(\mathcal{K}, H)| = O(n/\ell +  \ell  h^2\log h)$. Otherwise, the algorithm will return a separator $X\cup F(\mathcal{K}, H)$ at the iteration when the number of vertices in $H$ first falls below $2n/3$. The size is also $O(n/\ell +  \ell  h^2\log h)$ as claimed in \Cref{lm:Main}.

\begin{algorithm}[H]
\caption{$\textsc{BalancedSeparator}(G,\ell)$\Comment{\lightgray{Find a balanced separator of size $O(n/\ell + \ell h^2\log h)$}}}\label{alg:separator}
\begin{algorithmic}[1]
    \State $x\leftarrow$ arbitrary vertex of $G$
    \State  $\mathcal{K}\leftarrow \{x\}, \,\, X\leftarrow \varnothing, \,\, H \leftarrow$ largest connected component of $G\setminus x$
    \While{$|V(H)|\geq 2n/3$}
        \State $\Delta \leftarrow \ell \log h,\,\, S\leftarrow \textsc{LDD}(H,\Delta)$ \Comment{\lightgray{Low diameter decomposition, \Cref{cor:separation}}}\Comment{\blue{Step 1}}
        \If{every connected component of $H\setminus S$ has $\leq 2n/3$ vertices} % Step 1
            \State \Return $X\cup S\cup F(\mathcal{K},H)$  \Comment{\lightgray{$F(\mathcal{K},H)$ defined in \Cref{eq:FKH-de}}}
        \EndIf
        \State $v\leftarrow $ vertex such that $B_H(v,\Delta)\geq 2n/3$
        \State $\{L_0,L_1,\ldots\}\leftarrow$ layers of BFS tree $T$ in $H$ rooted at $v$
        \State $\ell^*\leftarrow (\log h+1)\ell, \,\, U\leftarrow L_{\leq \Delta},\,\, M \leftarrow L_{\Delta + 1: \Delta + \ell^*},\,\, B\leftarrow L_{\geq \Delta + \ell^* + 1}, \,\,M^{+\ell} \leftarrow M\cup L_{\Delta + \ell^* + 1:  \Delta + \ell^* + \ell}$. \Comment{\blue{End of Step 1}}
        \If{$N_H(C)\cap (U\cup M^{+\ell})\not= \varnothing$ for every $C\in \mathcal{K}$} \Comment{\blue{Step 2}}  % Step 2
           \State $x_C \leftarrow$ an arbitrary vertex in $N_H(C)\cap (U\cup M^{+\ell})$ for every $C\in \mathcal{K}$
           \State $\Tilde{C} \leftarrow \cup_{C\in \mathcal{K}} T[v,x_C]$ \Comment{\lightgray{$T[v,x_C]$ is the subpath of $T$ from $v$ to $x_C$}}
          \State $\mathcal{K}\leftarrow \mathcal{K}\cup \{\Tilde{C}\}$
           \State $H\leftarrow$ largest connected component of $H\setminus V(\Tilde{C})$. \Comment{\blue{End of Step 2}}
           
        \Else
           \State $C\leftarrow$ a branch set in $\mathcal{K}$ s.t. $N_H(C)\subseteq B\setminus M^{+\ell}$
            \If{$|B|\leq n/h$}  \Comment{\blue{Step 3}} % Step 3 
                \State $Y\leftarrow$ the layer of minimum size among layers in $M^{+\ell}\setminus M$; \,\, $y\leftarrow$ layer index of $Y$.
                \State $W\leftarrow  H[L_{\geq y+1}]$
                \State $Z\leftarrow$ vertices reachable from $N_H(C)$ in $W$
                \State $C\leftarrow C\cup Z$ 
                \State $H\leftarrow $ largest connected component of $H\setminus Z$. \Comment{\blue{End of Step 3}}
            \Else  \Comment{\blue{Step 4}}% Step 4 
                \State $i^*\leftarrow$ a layer in  $[\Delta+1, \Delta + \ell^{*}]$ s.t. $L_{i^*} \leq L_{\geq i^*+1}/\ell$ \Comment{\lightgray{$i^*$ exits by \Cref{lm:importantlayer}}}
                \State $X\leftarrow X\cup L_{i^*}$
                \State $H\leftarrow$  largest connected component of $H\setminus L_{\geq i^*}$  \Comment{\blue{End of Step 4}}
            \EndIf 
        \EndIf
        \State $\mathcal{K}\leftarrow \textsc{Trim}(\mathcal{K},H)$ \Comment{\lightgray{remove every branch set without a neighbor in $H$}}
    \EndWhile
    \State \Return $X\cup F(\mathcal{K}, H)$  \Comment{\lightgray{the separator when $|V(H)| < 2n/3$}}
\end{algorithmic}
\end{algorithm}

\begin{comment}
    
\begin{figure*}[ht!]
\centering\small
\begin{algorithm}
\textul{$\textsc{Separator}(G,\ell)$:} \+
\\ $\mathcal{K}\leftarrow \varnothing, S\leftarrow \varnothing, H \leftarrow G$
\\  \textbf{while} $|V(H)|\geq 2n/3$\+
\\ 
\\   repeat $n$ times \+
\\      $i\leftarrow$ a random point in $X$
\\      pick  random $t\in [0,1]$
\\      assign $a$ to $i$ if $y^*_{ai}\geq t$\-
\\ $D \leftarrow$ all points $j$ such that $B(j,\Delta/\beta)$ intersects more than 1 clusters $\{C_i\}$.
\\  $C'_i\leftarrow C_i\setminus D$
\\ $C''_{i} \leftarrow $ add to $C'_i$ all points $j$ such that $d_X(C'_i,j)\leq \Delta/(2\beta)$. 
\\ return $\{C''_i, \text{singletons}\}$
\end{algorithm}
\caption{The rounding scheme}
\label{alg:rounding}
\end{figure*}
\end{comment}

\subsection{The Analysis}

First, we show that the layer $i^*$ in line 24 exists. 

\begin{lemma}\label{lm:importantlayer}
In Step 4, there exists  $i^*\in [\Delta+1, \Delta + \ell^*]$ such that $L_{i^*} \leq L_{\geq i^*+1}/\ell$.    
\end{lemma}
\begin{proof}
Recall that in this case $|B|> n/h$.  Suppose otherwise, then for every $i\in [\Delta+1, \Delta + \ell^*], L_{i}\geq L_{\geq i+1}/\ell$. We have:
    \begin{equation*}
    \begin{split}
        L_{\geq \Delta+1}&=  L_{\Delta+1} +     L_{\geq \Delta+2} \geq (1+\frac{1}{\ell})    L_{\geq \Delta+2} \overset{\text{induction}}{\geq} (1+\frac{1}{\ell})^{\ell^*}  L_{\geq \Delta + \ell^*+1} =  (1+\frac{1}{\ell})^{\ell^*} |B|\\
        &> (1+\frac{1}{\ell})^{\ceil{\log(h)}\ell}\frac{n}{h}   \geq   2^{\ceil{\log(h)}} n/h\geq n
    \end{split}
    \end{equation*}
a contradiction.
\end{proof}

Next, we show several invariants maintained by the algorithm. Let $V(\cK) = \cup_{C\in \cK}V(C)$.

\begin{lemma}[Invariants]\label{lm:invariants} At the beginning of an iteration where $|V(H)|\geq 2n/3$, $\mathcal{K}, X$ and $H$ satisfy the following properties:
\begin{enumerate}
    \item  $\mathcal{K}$ is a clique minor of size $|\mathcal{K}|$, and hence $|\cK|\leq h-1$ in every iteration when $G$ is $K_h$-minor-free.
    \item   $H$ and branch sets in $\mathcal{K}$ are vertex-disjoint.
    \item  Every branch set $C\in \mathcal{K}$ has a neighbor in $H$. That is, $N_H(C) \not= \varnothing$.
    \item  For every branch set $C\in \mathcal{K}$, either $|V(C)| \leq (h-1)(2\log(h)+2) \ell$ or $|N_H(C)|\leq n/(h\ell)$.
    \item Let $\bdry H$ be the vertices in $H$ that are adjacent to vertices not in $H$. Then, $N_G(\bdry H)\setminus V(H) \subseteq X\cup V(\cK)$.
\end{enumerate}
\end{lemma}
\begin{proof}
  We prove by induction; the base case is when $\cK = \{x\}$, $X = \varnothing$, and $H$ is the largest connected component of $G\setminus \{x\}$. Then $F(\cK,H) = \{x\}$ by definition. Thus, all properties are trivially satisfied.

  Now we consider any iteration of the algorithm. In Step 1, if the algorithm returns a separator (in line 6), then it terminates, and all the properties follow from induction. Thus, we only consider Steps 2-4 where $H$ and $\cK$ are updated. Note that whenever the algorithm updates $\cK$ by either adding a new branch set (line 13) or growing an existing branch set in $\cK$ (line 21), it removes the vertices added to $\cK$ from $H$, and therefore $\cK$ in $H$ are always vertex disjoint, giving property 2. Observe that property 3 follows from line 27, where we remove from $\cK$ any branch set that has no neighbor in $H$. Property 5 follows from the fact that whenever we remove a set of vertices from $H$ (lines 14, 22, 26), these vertices are either added to a branch set in $\mathcal{K}$ or $X$. Note that in line 27, we only remove from $\cK$ branch sets that are not adjacent to $H$ and therefore, to $\bdry H$. It remains to show properties 1 and 4. 

  In Step 2, we add a new branch set $\Tilde{C}$ to $\mathcal{K}$. Lines 11 and 12 guarantee that $\Tilde{C}$ has an edge to every other branch set in $\cK$. Thus, property 1 holds.  To bound the size of $\Tilde{C}$, observe that $U\cup M^{+\ell}$ contains the first $\Delta + \ell^*+\ell = (2\log h + 2)\ell$ layers of the BFS tree rooted at $v$. Therefore, $|V(T[v,x_C])| \leq  (2\log h + 2)\ell$, giving $|V(\Tilde{C})|\leq |\cK| \cdot  (2\log h + 2)\ell \leq (h-1)  (2\log h + 2)\ell$. Thus, property 4 holds for $\Tilde{C}$.

  In both Steps 3 and 4, we do not add a new branch set to $\cK$ and therefore, property 1 follows directly from induction. Step 4 does not update $\cK$, and hence, we only need to verify property 4 after Step 3. Observe that $M^{+\ell}\setminus M$ has exactly $\ell$ layers. Thus, we have:
  \begin{equation*}
      |Y| \leq \frac{|M^{+\ell}\setminus M|}{\ell} \leq \frac{|B|}{\ell} \leq \frac{n}{h\ell}
  \end{equation*}
 as  $M^{+\ell}\setminus M\subseteq B$ and $|B|\leq n/h$. Since $W$ contains all the vertices in layers after $L_y = Y$ (line 19), the neighbors of $Z$ in the updated $H$ (after line 22) are in $Y$. Furthermore, since $N_H(C)\subseteq B\setminus M^{+\ell}$, $C$ has no neighbors in layers $L_{\leq y}$. Thus, in the updated $H$, $N_H(C)\subseteq Y$, giving $|N_H(C)| \leq |Y| \leq n/(h\ell)$ claimed in property 4.
\end{proof}

Now we are ready to prove \Cref{lm:Main}, which we restate below.

\mainLemma*
\begin{proof} The algorithm is in randomized polynomial time because of the LDD step. Otherwise, every step in each iteration can be implemented in a total of $O(m)$ time. Thus, the whole algorithm is in randomized polynomial time.

    Now we bound the size of the separator returned by \Cref{alg:separator}. In every iteration, we have:
    \begin{equation}\label{eq:FKH-size}
    \begin{split}
        F(\cK, H) &\leq \sum_{C\in \cK}\min\{|V(C)|, |N_H(C)|\} \\
        &\leq (h-1)\left((h-1)(2\log h+2)\ell + n/(h\ell)\right) \qquad \text{(by properties 1 and 4 in \Cref{lm:invariants})}\\
        &= O(n/\ell + \ell h^2\log h)
    \end{split}
    \end{equation}
    Observe that the algorithm only adds $L_{i^*}$ to $X$ in line 25, and there is no other update on $X$ in each iteration. We charge the size of $L_{i^*}$ to vertices in $L_{\geq i^*+1}$; every vertex is charged only $1/\ell$  $|L^*| \leq L_{\geq i^*+1}/\ell$ and is charged at most once during the entire algorithm, giving $|X|\leq n/\ell$. Thus, if the algorithm returns the separator in line 29, its size remains $O(n/\ell + \ell h^2\log h)$ by \Cref{eq:FKH-size}. Otherwise, it returns a separator in line 6, which adds a set $S$ of size $O(\log(h)n/(\ell \log h)) = O(n/\ell)$.  Thus, the total size of the separator remains in check.

    Finally, we show that the separator is balanced. Consider the last iteration where $|V(H)|\geq 2n/3$. For clarity, let $\Tilde{H}$ be the updated version of $H$ after this iteration; it holds that $|V(\Tilde{H})| < 2n/3$.  In Steps 1, 2, and 3, $\Tilde{H}$ is the largest connected component after we add a subset of vertices to the separator (in Step 3, the subset is $N_H(C)$). Since $|V(\Tilde{H})| < 2n/3$, every other connected component must have size at most $2n/3$ as well, and hence the separator is balanced. In Step 4, we only add $L_{i^*}$ to the separator, but we remove all vertices of $L_{\geq i^*}$. However, note that $|B_H(v,\Delta)|\geq 2n/3$ due to Step 1, at least  $2n/3$ vertices are in the first $\Delta$ layers, while vertices in $L_{\geq i^*}$ are in layers after $\Delta$. Thus, by adding $L_i^{*}$ to $X$, $\Tilde{H}$ is the same as the largest connected component of $H$ after removing only $L_i^{*}$ from $H$.   
\end{proof}

\paragraph{Towards a faster algorithm.}  While our algorithm described above calls LDD in every iteration, the whole analysis rests on the idea that there exists a vertex $v$ that contains $2n/3$ other vertices of $H$ within distance $\Delta = \ell \log(h)$. Instead of calling LDD to find such a vertex $v$, we sample a random vertex $s$, and with probability at least $2/3$, $s\in B_H(v, \Delta)$. Then, by the triangle inequality, every vertex in  $B_H(v, \Delta)\subseteq B_H(s, 2\Delta)$, and thus, we could use $s$ and distance threshold $2\Delta$ in place of $v$ and $\Delta$ in the above algorithm. Overall, we can replace all but one LDD calls by sampling and running BFS. The running time in each iteration becomes $O(m\log n)$, where the $\log n$  factor is to get a high probability. Observe that the number of iterations is  $O(h + n/(h\ell \log h) = O( h + \sqrt{n/\log h}) = O(\sqrt{n})$. Thus, the total running time is $O(m\sqrt{n}\log n)$ plus the running time of exactly one LDD call. The LDD in \Cref{cor:separation} is obtained by calling padded decomposition~\cite{CF25}, which can be implemented in $O(n^2)$ time. Implementing padded decomposition in $O(m)$ time (even in the case of constant $h$) remains an open problem. The LDD in \Cref{cor:separation} seems simpler than the padded decomposition and hence could be found faster than sampling from a padded decomposition. We leave this as an open problem for future work.

\section{A Flow-Based Approach}\label{sec:flow}

In this section, we construct a balanced separator for any $K_h$-minor-free graph of size $O(h\log h \sqrt{n})$ as in \Cref{thm:minor}. We adapt the technique of Korhonen and  Lokshtanov~\cite{KH24}, who constructed a balanced separator for induced-minor-free graphs based on the flow-cut duality~\cite{LR99,FHL05}. Here we observe that for  $K_h$-minor-free graphs, we can use the recent padded decomposition with a padding parameter $O(\log h)$ by Conroy and Filtser~\cite{CF25} to replace the $\log(n)$ factor with $\log(h)$ in the flow-cut duality. The rest of the proof is borrowed directly from Korhonen and  Lokshtanov~\cite{KH24}. We include all the details for completeness.

Let $\mathcal{P}_{st}(G)$ be the set of all simple paths between $s\not= t$ and $\mathcal{P}(G) = \cup_{(s,t)\in V\times V} \mathcal{P}_{st}(G)$  be the set of all simple paths in $G$. A \EMPH{concurrent flow} $\lambda: \mathcal{P}(G) \rightarrow \real_{\geq 0}$ is a function that assigns each path in $\mathcal{G}$ a non-negative value, a.k.a, the \EMPH{flow assigned to the path}, such that for ordered pair of vertices $(s,t)\in V\times V$, $\sum_{P\in \mathcal{P}_{st}(G)}\lambda(P) = 1$.  Thus, one can think of a concurrent flow as giving a distribution of paths between every pair of vertices. This is also the view we take in constructing a minor model when the algorithm fails to give a balanced separator. 

Let $\lambda(v)$ be the total amount of flow passing through $v$. More formally:
\begin{equation}
    \lambda(v)  = \sum_{P \in \mathcal{P}_{s,t}(G): ~v\in P,~ (s,t)\in V\times V} \lambda(P)
\end{equation}
The \EMPH{congestion} of the concurrent flow $\lambda$  is $\max_{v\in V}\lambda(v)$. The following lemma is a simple consequence of applying the padded decomposition for $K_h$-minor-free graphs with padding parameter $O(\log h)$ on top of the framework for vertex separator by Feige, Hajiaghayi, and Lee~\cite{FHL05}. %For completeness, we give a proof in \Cref{subsec:leighton-rao-proof}.

\begin{lemma}\label{lm:sep-vs-flow} Let $G$ be an $n$-vertex graph $K_h$-minor-free graph. For any parameter $\gamma > 0$,  there is a randomized polynomial-time algorithm that outputs either:
\begin{enumerate}
    \item  a balanced separator of $G$ of size $O(n^2\log h/\gamma)$, or
    \item  a concurrent flow of congestion $\gamma$ on  an induced subgraph $G'$ of $G$ such that $|V(G')|\geq 2n/3$.
\end{enumerate}
\end{lemma}

 The $\log(h)$ factor is best possible in \Cref{lm:sep-vs-flow}, and it factors directly into the size of the balanced separator. Thus, $O(h\log h \sqrt{n})$ is the barrier of the flow-based technique for balanced separators.

We will delay the proof of \Cref{lm:sep-vs-flow} to the end of this section, and continue with our algorithm for finding a balanced separator.  Let $\Ddot{H}$ be a graph obtained from $K_h$ by subdividing every edge twice and replacing each vertex of degree $h-1$ by a binary tree with $h-1$ leaves. Observe that if $G$ excludes $K_h$ as a minor, then $G$ also excludes $\Ddot{H}$ as a minor, e.g., Lemma 2 in \cite{BKLLM25}. We define an an \EMPH{almost minor embedding} of $\Ddot{H}$ to $G$ is a pair of mapping $(\phi, \pi)$ where $\phi: V(\Ddot{H})\rightarrow G$  and $\pi: E(\Ddot{H})\rightarrow \mathcal{P}(G)$ such that:

\begin{enumerate}
    \item for every edge $xy\in E(\Ddot{H})$,  $\pi(xy)$ is a simple $\phi(x)$-to-$\phi(y)$ path in $G$.
    \item for every two edges $xy$ and $uv$ in $\Ddot{H}$ that do not share any endpoint, $\pi(xy)\cap \pi(uv) =  \emptyset$. 
    \end{enumerate}

\FlowSeparator*
\begin{proof}
    Let $\gamma = \frac{n^{3/2}}{c\cdot h}$ for a sufficiently large constant $c$.  By applying the algorithm in \Cref{lm:sep-vs-flow} with parameter $\gamma$, we either find a balanced separator of size $O(h\log h \sqrt{n})$ of $G$ or a subgraph $G'$ of $G$ that has a concurrent flow of congestion $\gamma$. In the latter case, we can construct an almost minor embedding $\phi$ of $\Ddot{H}$ in $G'$, and therefore we can conclude that $G$ is not $K_h$-minor-free.  The construction is as follows:
  
\begin{itemize}
    \item  $\phi$ maps each vertex $u\in \Ddot{H}$ to a vertex $\phi(u)\in V(G')$ uniformly at random.
    \item  For every edge $uv\in \Ddot{H}$, $\pi(uv)$ will map $uv$ to a  path $P \in \mathcal{P}_{uv}(G')$ sampled from distribution induced by the concurrent flow $\lambda$. 
\end{itemize}  

Let $xy$ and $uv$ be two edges in $E(\Ddot{H})$ that do not share any endpoint. We say that $\pi(xy)$ and $\pi(uv)$ \EMPH{collide} if $\pi(xy)\cap \pi(uv) \not=\emptyset$. 

For every vertex $w$, we have:
\begin{equation}
    \begin{split}
        \Pr[w\in \pi(xy)]  &= \frac{1}{|V(G')|^2}\sum_{(a,b)\in V(G')\times V(G')}\sum_{w \in \mathcal{P}_{ab}(G')} \lambda(P)\\
        &= \frac{\lambda(w)}{|V(G')|^2} \leq \frac{9}{4n^2} \lambda(w) \leq \frac{9\gamma}{4n^2}
    \end{split}
\end{equation}
The penultimate inequality is due to $|V(G')|\geq 2n/3$, and the last inequality is by the definition of the congestion.   Since $xy$ and $uv$ do not share an endpoint,  $\pi(xy)$ and $\pi(uv)$ are two independent variables. Thus, we have:

\begin{equation}
     \Pr[w\in \pi(xy) \wedge w\in \pi(uv)]\leq \frac{81 \gamma^2}{16n^4}
\end{equation}
By taking the union bound over $w$, we have:

\begin{equation}\label{eq:colide-prob}
     \Pr[ \pi(xy) \text{ and } \pi(uv) \text{ collides}]\leq \frac{81 \gamma^2}{16n^3} 
\end{equation}

At this point, if we take the union bound over (at most) $h^4$ pairs of edges in $\Ddot{H}$, we would end up with the probability of having at least one collision is at most $\frac{81h^4\gamma^2}{16n^3}$. Thus, the probability that $\phi$ is an almost minor embedding of $\Ddot{H}$ is at least $1-\frac{81h^4\gamma^2}{16n^3}$, which is at least $1/2$ when $\gamma = O(n^{3/2}/h^2)$. However, we need  $\gamma =  O(n^{3/2}/h)$.

To this end, following Korhonen and  Lokshtanov~\cite{KH24}, we use the Lovász Local Lemma: Let $\mathcal{B}$ be a collection of bad events and let $p$ be such that $\max_{B\in \mathcal{B}}\Pr[B] \leq p$. Suppose that each bad event depends on at most $d$ other bad events. Then if $4dp\leq 1$, there is a non-zero probability that none of the bad events in $\mathcal{B}$ happens.

In our context, each bad event is $\pi(xy)\cap \pi(uv)\not = \emptyset$ for each non-adjacent pair.  The number of such pairs is $O(|E(\Ddot{H})|^2) = O(h^4)$. However, each bad event is only dependent on $O(|E(\Ddot{H})|) = O(h^2)$ other bad events since $\Ddot{H}$ has degree at most $3$. Thus, it suffices to have the probability in \Cref{eq:colide-prob} to be at most $1/(c'\cdot h^2)$ for a sufficiently large $c'$, giving $\gamma^{2}  = O(n^{3}/ h^2)$ and hence $\gamma \leq n^{3/2}/(c\cdot h)$, for some constant $c$, as desired. 
\end{proof}

It remains to prove \Cref{lm:sep-vs-flow}.  We denote a \EMPH{vertex cut} of $G$ by $(A,S,B)$, where (i) $A,B,S$ are vertex disjoint, (ii) $A\cup B\cup S = V$, and (iii) there is no edge between $A$ and $B$. Intuitively, $A$ and $B$ are obtained by removing $S$ from $G$. The \EMPH{sparsity} of $(A,S,B)$ is denoted by:
\begin{equation}
    \alpha_G(A,S,B) = \frac{|S|}{|A\cup S||B\cup S|}
\end{equation}

Let $\omega: V\rightarrow \real^+$ be a non-negative weight function on the vertex set of $G$. The weight function $\omega$ induces a metric on $V(G)$, denoted by $\langle G,\omega\rangle$, where the distance between any two vertices $u$ and $v$ is the shortest vertex-weighted distance, denoted by $\delta_{\omega}(u,v)$.  Let $f_{\omega}: V(G)\rightarrow \real$ be an embedding of $\langle G,\omega\rangle$ into the line; we say that $f_{\omega}$ is \EMPH{contracting} if $|f_{\omega}(u) - f_{\omega}(v)|\leq \delta_{\omega}(u,v)$ for every $u,v\in V$. We say that $G$ admits an embedding with an  \EMPH{average distortion} at most $D$ into the line if for every non-negative vertex weight function $\omega$, there exists a contracting embedding $f_{\omega}$ into $\real$ such that:

\begin{equation*}
    \sum_{(u,v)\in V\times V} |f_{\omega}(u) - f_{\omega}(v)|\geq \frac{1}{D} \sum_{(u,v)\in V\times V} \delta_{\omega}(u,v)
\end{equation*}

The following result by Feige, Hajiaghayi, and Lee~\cite{FHL05} relates embeddings into the line with low average distortion with sparsest vertex cut.

\begin{lemma}[Theorem 4.1.~\cite{FHL05}]\label{lm:avdsit-cut} If $G$ admits a low-distortion embedding into the line with distortion $D$ that is constructible in randomized polynomial time, then one can find in randomized polynomial time either: (1)  a concurrent flow with congestion $\gamma$ in $G$ or (2)  a vertex cut $(A,S,B)$ of sparsity $O(D/\gamma)$.
\end{lemma}

%\inlinenote{H: should introduce notation here.}

\begin{proof}[Proof of \Cref{lm:sep-vs-flow}]
Rabinovich~\cite{Rabinovich03} showed that if the shortest path metric (induced by edge weight function $w$) of $G$, denoted by $\delta_G(u,v)$, has a padded decomposition with padding parameter $\beta_G$, then the (edge-weight) metric $(V(G),\delta_G)$ admits an embedding into the line with average distortion $O(\beta_G)$. Observe that by setting $w(u,v) = \frac{\omega(u) + \omega(v)}{2}$, the shortest path metric w.r.t edge weight function $w$ is an embedding of the shortest path metric with vertex weight $\omega$ with distortion at most $2$. Furthermore, $\beta_G = O(\log h)$ for $K_h$-minor-free graphs~\cite{CF25}. Thus,  $\langle G,\omega\rangle$ is embeddable into the line with average distortion $O(\beta_G)$. The embedding can be constructed in randomized polynomial time w.h.p., as the padded decomposition is sampleable in polynomial time~\cite{CF25}. Thus, by \Cref{lm:avdsit-cut}, we either get a (1) concurrent flow with congestion $\gamma$ in $G$ or (2) a cut $(A,S,B)$ with sparsity $O(\log h/\gamma)$.

Note that the separator $S$ in the sparse cut $(A,S,B)$ might not be balanced. W.l.o.g, we assume that $A$ is the larger side, i.e., $|A| > 2n/3$. To get a balanced cut, the standard trick~\cite{FHL05,KH24} is to recurse on the subgraph induced by the larger side $G[A]$. We stop whenever we find a concurrent flow with congestion $\gamma$ in a subgraph of $G$ or the larger size of the cut has size at most $2n/3$.  Let $\{(A_i,S_i,B_i)\}_{i=1}^{t}$ be a sequence of sparse cuts in this process. In the latter case, the separator $\cup_{i=1}^t S_i$ is balanced and has size:

\begin{equation*}
    \begin{split}
        \sum_{i=1}^t |S_i| &= O(\log h/\gamma)\sum_{i=1}^t  |A_i||B_i| \\
        &\leq O(n\log h /\gamma) \sum_{i=1}^t  |B_i| \qquad \text{(as $|A_i|\leq n$)}\\
        &= O(n^2\log h/\gamma)
    \end{split}
\end{equation*}
since $\sum_{i=1}^t |B_i| \leq n$, concluding the proof.
\end{proof}

\paragraph{Acknowledgments.} We thank Jack Spalding-Jamieson for helpful conversations and comments.  The author is supported by NSF grant CCF-2517033 and NSF CAREER Award CCF-2237288.
\bibliographystyle{plain}
\bibliography{ref}

\end{document}